\documentclass[10pt]{article}
\usepackage{amsmath}
\usepackage{amsthm}
\usepackage{amsfonts}
\usepackage{amssymb}
\usepackage{graphics}
\usepackage{graphicx}
\usepackage{amsbsy}
\usepackage{verbatim}
\usepackage{enumerate}
\usepackage{amsmath}
\usepackage{amsthm}
\usepackage{amsfonts}
\usepackage{amssymb}
\usepackage{graphicx}
\usepackage{eepic}
\usepackage{epsfig} 
\usepackage{amsbsy}
\usepackage{enumerate}
\usepackage{verbatim}
\usepackage{epstopdf}
\newtheorem{remark}{Remark}[section]
\newtheorem{assumption}{Assumption}[section]
\newtheorem{theorem}{Theorem}[section]

\newtheorem{prop}{Proposition}[section]
\newtheorem{lemma}{Lemma}[section]
\newtheorem{definition}[theorem]{Definition}
\newcommand{\bx}{{\bf x}(t)}

\newcommand{\cK}{{\mathcal K}}

\newcommand{\bz}{{\bf z}}

\newcommand{\T}{{\mathbb T}}
\newcommand{\R}{{\mathbb R  }}

\newcommand{\cL}{\mathcal L}

\newcommand{\eps}{\epsilon}

\newcommand{\p}{_{\gamma\beta}}
\newcommand{\op}{\mathcal{L}}
\newcommand{\fspace}{L^2_{\rho}}

\usepackage[normalem]{ulem}

\begin{document}

\setlength{\baselineskip}{10pt}

\title{ ASYMPTOTIC ANALYSIS FOR THE GENERALIZED LANGEVIN EQUATION}
\author{M. Ottobre and  G.A. Pavliotis \\
        Department of Mathematics\\
    Imperial College London \\
        London SW7 2AZ, UK
                    }

\maketitle

\begin{abstract}
Various qualitative properties of solutions to the generalized Langevin equation (GLE) in a periodic or a confining potential are studied in this paper. We consider a class of quasi-Markovian GLEs, similar to the model that was introduced in~\cite{EPR99}. Geometric ergodicity, a homogenization theorem (invariance principle), short time asymptotics and the white noise limit are studied. Our proofs are based on a careful analysis of a hypoelliptic operator which is the generator of an auxiliary Markov process. Systematic use of the recently developed theory of hypocoercivity~\cite{Vil04HPI} is made.
\end{abstract}

%%%%%%%%%%%%%%%%%%%%%%%%%%%%%%%%%%%%%%%%%%%%%%%%%%%%%%%%%%%%%%%%%%%%%%%%%%%%

\section{Introduction}\label{sec:intro}

In this paper we study various qualitative properties of solutions to the generalized Langevin equation (GLE) in $\R^d$
\begin{equation}\label{e:gle}
\ddot{q} = -\nabla V(q) -\int_0^t \gamma(t-s) \dot{q}(s) \, ds + F(t),
\end{equation}
where $V(q)$ is a smooth potential (confining or periodic), $F(t)$ a mean zero stationary Gaussian process with autocorrelation function $\gamma(t)$, in accordance to the  fluctuation-dissipation theorem
\begin{equation}\label{e:fluct_dissip}
\langle F(t) \otimes F(s) \rangle = \beta^{-1} \gamma (t-s) I.
\end{equation}
Here $\beta$ stands for the inverse temperature and $I$ for the identity matrix. The GLE equation~\eqref{e:gle}, together with the fluctuation--dissipation theorem~\eqref{e:fluct_dissip} appears in various applications such as surface diffusion~\cite{Ferrando_al02} and polymer dynamics~\cite{snook07}. It also serves as one of the standard models of nonequilibirum statistical mechanics, describing the dynamics of a "small" Hamiltonian system (the distinguished particle) coupled to one or more heat baths which are modelled as linear wave equations with initial conditions which are distributed according to appropriate Gibbs measures~\cite{Rey-Bellet2006}. In this class of models the coupling between the distinguished particle and the heat bath is taken to be linear and is governed by a coupling function $\rho(x)$. The full Hamiltonian of the "particle+heat bath" model is
\begin{equation}\label{e:hamiltonian_full}
H(q,p, \phi, \pi) = H_{DP}(p,q) +  \mathcal{H}(\phi, \pi) + \lambda q \int \rho(x) \partial_q \phi(x) \, dx
\end{equation}
where $H_{DP}(q,p)$ denotes the Hamiltonian of the distinguished particle whose position and momentum are denoted by $q$ and $p$, respectively and $\mathcal{H}(\phi, \pi)$ is the Hamiltonian density of the wave equation where $\phi$ and $\pi$ are the canonically conjugate field variables. The linear coupling in~\eqref{e:hamiltonian_full} is motivated by the dipole approximation from classical electrodynamics. By integrating out the heat bath variables and using our assumptions on the initial conditions we obtain~\eqref{e:gle}, together with~\eqref{e:fluct_dissip}. The memory kernel $\gamma(t)$ in~\eqref{e:gle} is given by the coupling function through the formula
\begin{equation}\label{e:correlation}
\gamma(t) =  \int |\widehat{\rho}(k)|^2 e^{i k t} \, dk,
\end{equation}
where $\widehat{\rho}(k)$ denotes the Fourier transform of $\rho(x)$~\cite{JP97, Rey-Bellet2006}.

The GLE has also attracted attention in recent years in the context of mode reduction and coarse-graininig for high dimensional dynamical systems~\cite{GKS04}. One of the models that has been studied extensively within the framework of mode elimination is the Kac-Zwanzig model~\cite{FordKacMazur65, Zwanzig1973} and its variants~\cite{KSTT02, HK02, Kup03, ArielVandenEijnden09}. In this model, the heat bath is modelled as a finite dimensional system of $N$ harmonic oscillators with random frequencies and random initial conditions distributed according to a Gibbs distribution at inverse temperature $\beta$. The heat bath can be coupled either linearly or nonlinearly with the distinguished particle~\cite{KS03}. Just as with model~\eqref{e:hamiltonian_full}, we can integrate out the heat bath variables explicitly. Passing then to the thermodynamic limit $N \rightarrow +\infty$, we obtain the GLE~\eqref{e:gle}. The form of the memory kernel $\gamma (t)$ depends on the choice of the distribution of the spring constants of the harmonic oscillators in the heat bath~\cite{GKS04}. The Kac-Zwanzig model and its variants have proved to be very useful for testing various methodologies and techniques such as transition state theory~\cite{ArielVandenEijnden07,HanTalkBork90}.

The GLE~\eqref{e:gle} is a stochastic integrodifferential equation which is equivalent to the original infinite dimensional Hamiltonian system with random initial conditions. The infinite dimensionality of the original Hamiltonian dynamics with random initial conditions (or, equivalently, the non-Markovianity of the finite dimensional stochastic dynamics~\eqref{e:gle}) renders the analysis of this dynamical system very difficult. This problem was studied in detail by Jaksic and Pillet in a series of papers~\cite{JP97, JaksicPillet97, JaksicPillet98}. In these works, existence and uniqueness of solutions as well as the ergodic properties of~\eqref{e:gle} were studied in detail. In particular, it was shown that the process $\{ q, p = \dot{q} \}$ is mixing with respect to the measure
$$
\mu_{\beta} (dq dp) = \frac{1}{Z_{\beta}} e^{-\beta H_{DP}(q,p)} \, dq dp.
$$
To our knowledge, no information concerning the rate of convergence to equilibrium for the non-Markovian dynamics~\eqref{e:gle} is known for general classes of memory kernels. Ergodic theory for a quite general class of non-Markovian processes has been developed recently, see~\cite{Hairer2009} and the references therein.

A class of memory kernels for which more detailed information on the long time asymptotics of the GLE~\eqref{e:gle} can be obtained was considered by Eckmann, Hairer, Pillet and Rey-Bellet in a series of papers~\cite{ReyBelletThomas2002,EckmPillR-B00, EPR99, EH00}. Based on a generalization of Doob's theorem on stationary, Markovian, Gaussian processes~\cite{DymMcKean1976}, it was observed in these works that when the memory kernel $\gamma(t)$ has a rational spectral density, then the GLE~\eqref{e:gle} is equivalent to a {\bf finite dimensional} Markovian system. This system is obtained by adding a finite number of additional degrees of freedom which account for the memory in the system. These auxiliary variables satisfy linear stochastic differential equations. As an example we mention the case where $\widehat{\rho}(k)$ in~\eqref{e:correlation} can be written as
$$
|\widehat{\rho}(k)|^2 = \frac{1}{|p(k)|^2}
$$ 
where $p(k) = \sum_{m=1}^M c_m (- k)^m $ is a polynomial with real coefficients and roots in the upper half plane. Then the Gaussian process with spectral density $|\widehat{\rho}(k)|^2$ is the solution of the SDE
$$
p \left(-i \frac{d}{dt} x(t) \right) = \frac{d W}{d t},
$$
where $W(t)$ is a standard one dimensional Brownian motion.
A related finite dimensional approximation of the infinite dimensional dynamics~\eqref{e:gle} has been introduced by Mori~\cite{Mori65}, see also~\cite{Grigolini1982} and the references therein. Mori's technique is based on a continued fraction expansion of the Laplace transform of the memory function $\gamma(t)$. 

Motivated by the above, in this paper we will consider finite dimensional approximations of the GLE. The general form of the Markovian approximation of~\eqref{e:gle} can be written as~\cite{Kup03}~\footnote{To simplify the notation we consider~\eqref{e:gle} in one dimension. The results presented in this paper are valid in arbitrary finite dimensions. More details on the notation and on the multidimensional case can be found in Section~\ref{sec:notation} and Remark~\ref{rem:m,d>1}.}
\begin{subequations}\label{e:markovian}
\begin{eqnarray}
\dot{Q}_m(t) & = & P_m(t), \quad Q_m (0) = q(0), \\
\dot{P}_m(t) & = & - \nabla V (Q_m(t)) + \lambda^T z(t), \quad P_m (0) = p(0), \\
\dot{z}(t)   & = & - P_m(t) \lambda - A z(t) + C \dot{W}, \quad z(0) \sim \mathcal{N}(0, I ),
\end{eqnarray}
\end{subequations}
where $z: \R^+ \mapsto \R^m$ and $\lambda \in \R^m, \; A, \, C \in \R^{m \times m}$. The fluctuation-dissipation theorem, which takes the form $C C^T = \beta^{-1} (A + A^T)$, is assumed to be satisfied. 

In this paper we will consider~\eqref{e:markovian} with $\lambda = (\lambda_1, \lambda_2, \dots ,\lambda_m)$ and $A$ diagonal with $A_{ii} = \alpha_i >0$. This amounts to approximating the memory kernel by a sum of exponentials,
\begin{equation}\label{e:memory_approx}
\gamma_m(t) = \sum_{i=1}^m \lambda_i^2 e^{- \alpha_i |t|} .
\end{equation} 
It is expected that the results proved in this paper are also valid in the more general case~\eqref{e:markovian}. As remarked in~\cite{Kup03}, when $A$ is invertible, the more standard Mori approximation~\cite{Mori65} is equivalent to~\eqref{e:memory_approx} after an appropriate orthogonal transformation. 

For this particular choice of $\lambda$ and $A$ the SDEs~\eqref{e:markovian} become (we drop the subscripts $m$ for notational simplicity)
\begin{subequations}\label{syst_n_additional_proc}
\begin{eqnarray}
\dot{Q}(t) & = & P(t), \quad Q (0) = q(0), \\
\dot{P}(t) & = & - \nabla V(Q(t)) + \sum_{j=1}^m \lambda_j z(t), \quad P (0) = p(0), \\
\dot{z}_j(t)   & = & - \lambda_j P(t)  - \alpha_j z(t) + \sqrt{2 \alpha_j \beta^{-1}} \dot{W}_j,  \quad z_j(0) \sim \mathcal{N}(0, \beta^{-1} )
\end{eqnarray}
\end{subequations}
for $j=1, \dots , m$. The process $\{Q(t), \, P(t), \, z(t) \}$ is Markovian with generator $-\cL$ given by
\begin{eqnarray}
-\cL &=& p \cdot \nabla_q -\nabla_q V(q) \cdot \nabla_p +  \sum_{j=1}^m \lambda_j z(t) \cdot \nabla_p 
\nonumber \\ && 
+ \sum_{j=1}^m \lambda_j \left(- \lambda_j p \nabla_{z_j}  - \alpha_j z_j \cdot \nabla_{z_j} + \alpha_j \beta^{-1} \Delta_{z_j} \right).  \label{e:generator}
\end{eqnarray}
This is a degenerate second order elliptic differential operator of hypoelliptic type~\cite{Hormander1967}. Convergence to equilibrium for models of the form~\eqref{syst_n_additional_proc} has been studied using functional analytic techniques~\cite{EPR99, EH00}. Similar results have also been proved using Markov chain techniques~\cite{MatSt02, ReyBelletThomas2002}. In this paper we present an alternative proof of exponentially fast convergence to equilibrium in relative entropy using the recently developed theory of hypocoercivity~\cite{Vil04HPI}. 

Several other results are also proved in this paper. We obtain sharp estimates on the derivatives of the Markov semigroup generated by $-\cL$. We prove a homogenization theorem (invariance principle) for~\eqref{syst_n_additional_proc} in a periodic potential and we obtain estimates on the diffusion coefficient. Finally, we study the white noise limit of the GLE~\eqref{e:gle}, i.e. the limit as the noise $F(t)$ in~\eqref{e:gle} converges to a white noise process. We show that in this limit the solution of~\eqref{syst_n_additional_proc} converges strongly to the solution of the Langevin equation
\begin{equation}\label{e:langevin}
\ddot{q} = -V'(q) -\gamma \dot{q} +\sqrt{2 \gamma \beta^{-1}} \dot{W}
\end{equation}
and we obtain a formula for the friction coefficient $\gamma$ in terms of the coefficients $\{ \lambda_j, \, \alpha_j \}_{j=1}^m$.

The rest of the paper is organized as follows. In Section~\ref{sec:results} we state our main results and we introduce the notation that we will be using. In Section~\ref{sec:equilibrium} we prove exponentially fast convergence to equilibrium in relative entropy. In Section~\ref{sec:derivatives} we prove estimates on the derivatives of the Markov Semigroup generated by $-\cL$. In Section~\ref{sec:homogen} we prove the homogenization theorem. In Section~\ref{sec:markovian} we study the white noise limit. For the reader's convenience, background material on the theory of hypocoercivity is summarized in Appendix~\ref{app:hypoco}. Finally, the proof of geometric ergodicity of the process~\eqref{syst_n_additional_proc} using Markov chain techniques is presented in Appendix~\ref{sec:app_ergodicity}.

%%%%%%%%%%%%%%%%%%%%%%%%%%%%%%%%%%%%%%%%%%%%%%%%%%%%%%%%%%%%%%%%%%%%%%%%%%%%%

%
%%%%%%%%%%%%%%%%%%%%%%%%%%%%%%%%%%%%%%%%%%%%%%%%%%%%%%%%%%%%%%%%%%%%%%%%%%%%%%%%%%%%%%
%
\section{Statement of Main Results}
\label{sec:results}
We will use the notation $X:= \T^{d} \times \R^{d} \times \R^{dm}$ and $Y:= \R^{d} \times \R^{d} \times \R^{dm}$. We will also denote the process $\{q(t), \, p(t), \, \bz(t) \}$ by $\bx$. When we study the dynamics~\eqref{syst_n_additional_proc} in $X$ the potential $V(q)$ is periodic, whereas when $\bx \in Y$ the potential will be taken to be confining. The precise assumptions on the potential are given in Assumption~\ref{assump:potential} below.

Our fist result concerns the ergodicity of the SDE~\eqref{syst_n_additional_proc} in  $X$ or in $Y$. To prove the ergodicity of the SDE in $Y$ we need to make the following assumptions on the potential.

\begin{assumption}\label{assump:potential}
\item (i)$\; V(q)\in C^2(\R^d)$ is a confining potential. 
\item(ii) $ \langle \nabla_qV,q \rangle \geq \sigma V(q)+\beta \|q\|^2 \;\mbox{for some suitable} \;\beta,\sigma>0$ where $\langle  \cdot, \cdot \rangle$ and $\|\cdot \|$ denote the Euclidean inner product and norm, respectively.
\item (iii) There exists a constant $c$ such that  $\| \nabla^2V \| \leq c $, where $\| \cdot \|$ denotes the Frobenious-Perron matrix norm.
%\item(iii) $ \tilde{\sigma}\|q\|^2-\tilde{\beta}\|\nabla_qV\|^2\rightarrow +\infty \;\mbox{as }\|q\|^2\rightarrow +\infty .$
\end{assumption}

The invariant measure $\mu_{\beta}(dq \, dp \, d \bz) = \rho(p,q, \bz) \, dp dq d \bz$ of the process~\eqref{syst_n_additional_proc}, whose density satisfies the stationary Fokker-Planck equation $\cL^* \rho =0$, is known:
\begin{equation}\label{e:gibbs_gle}
\rho (q,p, \bz) = \frac{1}{Z} e^{-\beta( \frac{1}{2} |p|^2 + V(q) + \frac{1}{2}\|\bz\|^2)},
\end{equation}
where $\mathcal{Z}$ is the normalization constant. This invariant measure is unique and the law of the process~\eqref{syst_n_additional_proc} converges exponentially fast to $\mu_{\beta}$. 

\begin{theorem}[Ergodicity] The process (\ref{syst_n_additional_proc}) with $\bx \in X$ and $V(q)\in C^3(\mathbb{T}^d)$ is geometrically ergodic. The same holds true when $\bx \in Y$, provided that the potential $V(q)\in C^3(\R^d)$ satisfies Assumption~\ref{assump:potential}.
\end{theorem}

The proof of this theorem, which is based on Markov chain-type arguments and which is similar to the proof presented in~\cite{ReyBelletThomas2002}, see also~\cite{MatSt02}, can be found in Appendix~\ref{sec:app_ergodicity}. 

We can prove exponentially fast convergence to equilibrium using tools from the theory of hypocoercivity~\cite{Vil04HPI}. We will use the notation $\cK :=\mbox{Ker}(\cL)$ and $H^1_{\rho}$ for the weighted Sobolev space  $H^1$ with respect to $\mu_\beta$ on either $X$ or $Y$.  

\begin{theorem}\label{thm:converg_equil}
Let $-\cL$ be the generator of the process $\bx \in X$, the solution of~\eqref{syst_n_additional_proc}. Then there exist constants $C, \, \lambda >0$ such that
$$
\|e^{-t \cL} \|_{H_{\rho}^1/\mathcal{K}\ \rightarrow H_{\rho}^1/\mathcal{K}} \leq C e^{- \lambda t}.
$$
The same holds true when $\bx \in Y,$ provided that the potential  satisfies Assumption~\ref{assump:potential}(i) and (iii).

\end{theorem}

Using the tools from~\cite{Vil04HPI} we can prove exponentially fast convergence to equilibrium in relative entropy. The relative entropy (or Kullback Information) between two probability measures $\mu$ and $\nu$ with smooth densities $f$ and $\rho$, respectively, is defined as
$$
H_{\rho}( f) = \int f \log \left(\frac{f}{\rho}  \right) \, d{\bf x}. 
$$  
We will measure the distance between the law of the process $\bx$ at time $t$ and the equilibrium distribution in relative entropy.

\begin{theorem}[Convergence to Equilibrium]\label{thm:converg_rel_entropy} 
Let $f_t$ be the law of the process $\bx$ at time $t$ and assume that $H_{\rho}(f_0) < +\infty$ and $V(q) \in C^2(\T^d)$. Then there exist constants $C, \, \alpha >0$ such that
$$
H_{\rho}(f_t) \leq C e^{-\alpha t}H_{\rho}(f_0).
$$
The same holds true when $\bx  \in Y$, assuming that $H_{\rho}(f_0) < +\infty$ and provided that the potential $V(q)$ satisfies Assumption~\ref{assump:potential}(i) and (iii). 

\end{theorem}

\begin{remark}
In view of the Kullback inequality
\begin{equation}\label{Kullback inequalty}
\frac{1}{2}\|f_t -\rho \|_{L^1}   \leq H_{\rho}(f_t),
\end{equation} 
Theorem~\ref{thm:converg_rel_entropy} implies that, for initial data with finite relative entropy, we have exponentially fast convergence to equilibrium in $L^1$.
\end{remark}

The proofs of Theorems~\ref{thm:converg_equil} and~\ref{thm:converg_rel_entropy} are presented in Section~\ref{sec:equilibrium}.

Estimates on the Markov semigroup and its derivatives associated to the Langevin equation can be proved using an appropriate Lyapunov function~\cite{Herau2007,HP07}. In this paper we use similar techniques to obtain estimates on the Markov semigroup and its derivatives for the generalized Langevin equation, equation~\eqref{syst_n_additional_proc}. We introduce $C_k, \, k=0,1,2$ with $C_0 = A, \, C_1 =[A,B]$ and $C_2 = [C_1,B]$. We will use the notation $L^2_{\rho} := L^2 ( \, \cdot \, ; \mu_{\beta} (d {\bf x}))$ where $\cdot$ is either $X$ or $Y$.

\begin{theorem}[Estimates on Derivatives of the Markov Semigroup]\label{thm:short_time_asymp}
Let $-\cL$ be the generator of the process $\bx \in X$, the solution of~\eqref{syst_n_additional_proc} with $V(q) \in C^2(\T^d)$. Then the Markov semigroup $e^{-t \cL}$ satisfies the bounds
\begin{eqnarray}\label{e:short_time}
\|C_k e^{-t\op} \|_{L^2_{\rho} \rightarrow L^2_{\rho}} \leq C   \frac{1}{t^{\frac{1 + 2 k}{2}}}, \quad k=0,1,2 \mbox{ and } t\in(0,1].
\end{eqnarray}
The same holds true when $\bx \in Y$, provided that the potential $V(q)$ satisfies Assumption~\ref{assump:potential}(i) and (iii). 
\end{theorem}
\begin{remark}
This result can also be obtained by applying Theorem \ref{thm:hypoeelipticity short time asymp}. Malliavin calculus-type arguments show that estimate (\ref{e:short_time}) is sharp.     
\end{remark}

When the potential $V(q)$ is periodic, the particle position, appropriately rescaled, converges weakly to a Brownian motion with a diffusion coefficient which can be calculated in terms of the solution of an appropriate Poisson equation. Results of this form have been known for a long time for the Smoluchowski (overdamped) equation~\cite[Ch. 13]{PavlSt08} as well as for the Langevin dynamics~\cite{papan_varadhan, HairPavl04}. In this paper we prove a similar result for the generalized Langevin equation. We will use the notation $\phi^e := \phi \cdot e, \, p^e :=p \cdot e$,  where $e$ denotes an arbitrary unit vector in $\R^d$.

\begin{theorem}[Homogenization]\label{thm:homogenization}
Let $\bx \in X$ be the solution of~\eqref{syst_n_additional_proc} with $V(q) \in C^{\infty}(\T^d)$ with stationary initial conditions. Then the rescaled process $q^e_\eps (t):=e \cdot\eps q(t/\eps^2)$ converges weakly on $C([0,T],\R)$ to a Brownian motion with diffusion coefficient $D$ with 
\begin{equation}\label{e:deff}
D^e := D e \cdot e = \beta^{-1} \sum_{j=1}^m \alpha_j \|\partial_{z_j} \phi^e \|^2,
\end{equation}
where $\phi^e \in \fspace$ is the unique, smooth, mean zero, periodic in $q$ solution of the Poisson equation
\begin{equation}\label{e:poisson_e}
\cL \phi^e =p^e
\end{equation}
on $X$. Furthermore, the following estimates hold
\begin{equation}\label{e:deff_estim}
0 < D^e \leq  \frac{4}{\beta}\sum_{i=1}^m\frac{\alpha_i}{\lambda_i^2}.
\end{equation}
\end{theorem}

Let $q(t)$ be the solution of the Langevin equation~\eqref{e:langevin} and let $q^{\gamma}(t):=q(\gamma t)$. It is well known that this rescaled process converges
 in the overdamped limit $\gamma \rightarrow +\infty$ to the solution of the Smoluchowski equation~\cite[Ch. 10]{nelson}
\begin{equation}\label{e:smoluchowski}
\dot{q} = -\nabla V(q) +\sqrt{2 \beta^{-1} } \dot{W}.
\end{equation}
Similar results have also been proved in infinite dimensions~\cite{CerFreid06a}. In this paper we prove a similar result for the convergence of solutions to the GLE to the Langevin equation in the strong topology and obtain a formula for the friction coefficient that appears in the Langevin equation.

Consider~\eqref{e:gle} with the rescaled noise process
\begin{equation}\label{e:rescaled_noise}
F^{\eps}(t) :=\frac{1}{\sqrt{\eps}} F(t/ \eps),
\end{equation}
which is a mean zero stationary Gaussian process with autocorrelation function
\begin{equation}\label{e:rescaled_memory}
\gamma^{\eps}(t) = \frac{1}{\eps} \gamma (t/ \eps).
\end{equation}
For the memory kernel~\eqref{e:memory_approx}, $\gamma^{\eps}(t)$ becomes
\begin{equation}\label{e:rescaled_memory_1}
\gamma^{\eps}(t) = \sum_{j=1}^m \frac{\lambda_j^2}{\eps} e^{- \frac{\alpha_j}{\eps}|t|}.
\end{equation}
Consequently, the rescaled noise process~\eqref{e:rescaled_noise} is obtained by rescaling the coefficients in~\eqref{syst_n_additional_proc} according to  $\lambda_j \rightarrow \frac{\lambda_j}{\sqrt{\eps}}$, $\alpha_j \rightarrow \frac{\alpha_j}{\eps}$. Under this rescaling the SDEs become 
%considering such a scaling is equivalent to scaling time, $q(t/\epsilon)$, and then taking the weak coupling $\lambda_i\rightarrow\sqrt{\epsilon}\lambda_i$. Putting $\beta=1$,  system ~\eqref{e:markovian} becomes
\begin{equation}\label{white noise system}
\left\{
\begin{array}{ccl}
\dot{q}(t)&=&p(t),\\
\\
\dot{p}(t)&=&-\partial_qV(q)+
\sum_{i=1}^m\frac{\lambda_i}{\sqrt{\epsilon}}z_i(t),\\
\\
\dot{z}_i(t)&=&-\frac{\lambda_i}{\sqrt{\epsilon}}p_t-\frac{\alpha_i}{\epsilon}z_i(t)+
\sqrt{\frac{2\alpha_i \beta^{-1}}{\epsilon}}\dot{W}_i, \quad i=1,...,m.
\end{array}
\right.
\end{equation} 

\begin{theorem}[The White Noise Limit]\label{thm:markovian limit}
Let $\{q(t), p(t), \bz(t) \} \in X$ be the solution of (\ref{white noise system}) with $V(q) \in C^{1}(\T^{d})$ and initial conditions having finite moments of any order. Then $\{q (t), p (t) \}$ converge strongly to the solution of the Langevin equation
\begin{equation}\label{limit system for white noise scaling}
\left\{
\begin{array}{ccl}
\dot{Q}(t)&=&P(t),\\
\dot{P}(t)&=&-\nabla_qV(Q_t)-
\sum_{i=1}^m\left(\frac{\lambda_i^2}{\alpha_i} P(t)-\sqrt{\frac{2\lambda_i^2}{\alpha_i \beta^{-1}}}
\dot{W}_i\right).\\
\end{array}
\right.
\end{equation}
Consequently, the process $ \{ q,p \}$ converges weakly to the solution of the Langevin equation
\begin{equation}\label{e:sde_white_noise}
\left\{
\begin{array}{ccl}
\dot{Q}(t)&=&P(t),\\
\dot{P}(t)&=&-\nabla_q V(Q_t)-\gamma P+\sqrt{2\gamma \beta^{-1}}\dot{W},\\
\end{array}
\right.
\end{equation} 
where the friction coefficient $\gamma$ is given by the formula 
$$
\gamma=\sum_{j=1}^m\frac{\lambda_i^2}{\alpha_i}.
$$
\end{theorem}
%\begin{equation}\label{e:friction}
%\gamma = \sum_{j=1}^N\frac{\lambda_j^2}{\alpha_j}.
%\end{equation}
%\end{theorem}
% 
%%%%%%%%%%%%%%%%%%%%%%%%%%%%%%%%%%%%%%%%%%%%%%%%%%%%%%%%%%%%%%%%%%%%%%%%%%%%%%%%%%%
%
\subsection{Notation}\label{sec:notation}
For $\bx=(q,p,{\bf z})\in Y:=\mathbb{R}^d\times \mathbb{R}^d\times \mathbb{R}^{dm}\, $ or $\bx\in X:=\mathbb{\T}^d\times \mathbb{R}^d\times \mathbb{R}^{dm}\, $ consider the operator $\op$:
\begin{eqnarray}
-\op & = &  p \nabla_q - \nabla_qV(q) \nabla_p + \left(\sum_{j=1}^m \lambda_j z_j \right) \nabla_{p} 
\nonumber\\  &&
+  \sum_{j=1}^m \left(- \alpha_j z_j \nabla_{z_j} - \lambda_j p \nabla_{z_j} + \beta^{-1} \alpha_j \nabla^2_{z_j} \right), \label{operator}
\end{eqnarray}
with kernel $\mathcal{K}:=Ker\op$. The density of the invariant measure of the process $\bx$ is
\begin{equation}\label{equilibrium measure}
\rho(p,q,z)=\frac{1}{\mathcal{Z}}e^{-\beta(V(q)+\frac{1}{2}\mid p\mid^2+\frac{1}{2}\mid z\mid^2)}, \quad \mathcal{Z}=\int e^{-\beta(V(q)+\frac{1}{2}\mid p\mid^2+\frac{1}{2}\mid z\mid^2)} dpdqdz,
\end{equation}
where $\mid \cdot\mid$ denotes either the Euclidean or the matrix norm. In (\ref{operator}), $\nabla$ is the gradient (or the derivative when $d=1$) and $\Delta$ the Laplacian. $\nabla^2$ denotes the Hessian and if $O$ is an operator then $O^*$ is its adjoint in $\fspace := L^2 (\, \cdot \, ; \mu_{\beta} (d {\bf x}))$. Furthermore, we will use the notation $L=-\op$ so that $L$ is the actual generator of the process.\\ 
Define
\begin{equation}\label{def of B}
B=-p\nabla_q + \nabla_q V\nabla_p- \sum_{j=1}^m \lambda_j \left( z_j \nabla_{p}  -p \nabla_{z_j}\right).
\end{equation}
We easily check that $B^*=-B.$
If  $m=1$ then $A_i=-\partial_{z_i}$ (derivative with respect to the $i-th$ component of $z$) so that $ A^*_i=-z_i+\partial_{z_i}$
and we can write
\begin{equation}\label{L=A*A+B}
\op=B+\sum_{i=1}^{d}A_i^*A_i=B+A^*A,
\end{equation}
where $A$ is intended to be the row vector of operators $(A_1,...,A_d)$ ( the same for $A^*$). More precisely, if $m=1$ then: 
$A:\fspace\longrightarrow\fspace \otimes \mathbb{R}^d$,  
 $B:\fspace\longrightarrow \fspace$, 
$[A^*,A]:\fspace \longrightarrow \fspace$, being $[A^*,A]:=\sum_{j=1}^d [A^*_j,A_j]$; on the other hand
$[A,A^*]:\fspace\longrightarrow \fspace\otimes \mathbb{R}^d\otimes \mathbb{R}^d$ is a matrix whose $ij$-th component is given by $ [A,A^*]_{ij}:=[A_i,A_j^*]$; in an analogous way 
$[A,A]: \fspace\longrightarrow \fspace\otimes \mathbb{R}^d\otimes \mathbb{R}^d$ is a matrix with $[A,A]_{ij}:=[A_i,A_j]$; finally $C:=[A,B]$,
$C:\fspace\longrightarrow \fspace\otimes \mathbb{R}^d $ is a vector of operators,    $C_i=[A_i,B], \,i=1...d$,  and the same holds for $C_2:=[C,B]$, $C_2:\fspace\longrightarrow \fspace\otimes \mathbb{R}^d$.
  
If $m>1$ then  (\ref{L=A*A+B}) becomes
\begin{equation}\label{l=b+A*A general}
\op=B+\sum_{i=1}^{m}\sum_{j=1}^dA_{ij}^*A_{ij}
\end{equation}
with $A_{ij}=-\partial_{z_{i_j}}$ i.e. derivative with respect to the $j$-th component of $z_i.$
We will use the notation
\begin{equation}\label{L=B+A^*A syntetic}
\op=B+A^*A, 
\end{equation}
meaning either (\ref{L=A*A+B}) or (\ref{l=b+A*A general}).\\
As for the norms, unless otherwise specified, 
$\| \cdot \| $ indicates the norm  of $\fspace$, $\|\cdot \|_1^2=
 \|A\cdot\|^2+\|C\cdot\|^2+\|C_2\cdot\|^2$ is a sort of homogeneous $H^1(Y ; \mu_{\beta}(d {\bf x}))$ norm and 
$\|\cdot\|_{H^1}^2=\|\cdot\|^2+\|A\cdot\|^2+\|C\cdot\|^2+\|C_2\cdot\|^2$
is the usual inhomogeneous one. The inner products in these Hilbert spaces  are denoted by $(\cdot ,\cdot )$, $(\cdot ,\cdot )_1$ and $( \cdot,\cdot )_{H^1}$, respectively. 
\begin{remark}\label{rem:maximal accretivity}
It is a general result that if $\op$ is \textit{accretive} and  $\bar{\op}$ is maximally accretive then $\bar{\op}$ is the generator of a contraction semigroup, so $e^{-\bar{\op} t}$ is well defined~\cite[Ch. 5]{HelNi05}.\\
$\op$ is accretive since, as already noticed, $( \op f,f)=\|Af \|^2\geq 0 \; \forall f$ in the domain of $\op$.
To prove that  $\bar{\op}$ is maximally accretive we show that the operator $ 
\op^{\ast}+\nu I$ is injective for all $\nu >0$. Since the operator is linear, we just need to show that $Ker (\op^{\ast}+\nu I)= \{0 \}$ for all $\nu >0$. 
So suppose $f\in Ker \big( \op^{\ast}+\nu I \big)$, then
$$
\op^{\ast}f+\nu f=-Bf+A^{\ast}Af+\nu f=0.
$$ 
Consequently
$$
(-Bf,f)+\|Af\|^2+\nu\|f\|^2=0.
$$
Therefore
$$
0\leq\|Af\|^2=-\nu\|f\|^2\Rightarrow f=0.
$$
Notice that when $\nu=1$, from the above we conclude that $A^{\ast}A$ is self-adjoint and not only symmetric~\cite[Thm VIII.3]{reed-simon}. From now on by $\op$ we mean its closure.
\end{remark}

%%%%%%%%%%%%%%%%%%%%%%%%%%%%%%%%%%%%%%%%%%%%%%%%%%%%%%%%%%%%%%%%%%%%%%%%%%%%%%%%%%%%%%%5
\section{Convergence to Equilibrium}
\label{sec:equilibrium}
\subsection{Hypocoercivity}

Background material on hypocoercivity is presented in Appendix~\ref{app:hypoco}.

\begin{definition}[Hypocoercivity]
With the same notation as in Definition \ref{def:coercivity} and assuming that the operator $\mathcal{T}$ generates a continuous semigroup, such an operator is said to be $\lambda$-hypocoercive on $\tilde{\mathcal{H}}$ if there exists a constant $\kappa>0$ such that
$$
\parallel e^{-\mathcal{T}t}h\parallel_{\tilde{\mathcal{H}}}\leq \kappa e^{-\lambda t}\parallel h\parallel_{\tilde{\mathcal{H}}}
\quad \forall h\in \tilde{\mathcal{H}} \mbox{ and }t\geq 0.
$$   
\end{definition}
We say that an unbounded linear operator $S$ on $\mathcal{H}$ is \textit{relatively bounded} with respect to the (linear unbounded) operators $T_1,...,T_n$ if $\mathcal{D}(S)\subset (\cap\mathcal{D}(T_j)$) and $\exists$ a constant $\alpha>0$ s.t.
$$
\forall h\in \mathcal{D}(S),\qquad \|Sh\|\leq \alpha (\|T_1h\|+...+\|T_nh\|).
$$
The basic idea employed in the the theorems that we are going to use is to appropriately construct a scalar product on $H^1_{\rho}$ by adding lower order terms and then use the fact that hypocoercivity is invariant with respect to a change of equivalent norms, whereas coercivity does not enjoy such invariance.
Finally, notice that 
$\mathcal{S},$ the class of Schwartz functions is dense in $\fspace$, hence it is dense in $D(A)\cap D(B)$. This guarantees that all the operations performed on these (unbounded) operators are well defined.\\
Set $m=1=d$, $\alpha=\lambda=\beta =1$. The first two commutators are
\begin{equation}
C_1 = C = [A,B] = \partial_p \quad \mbox{and} \quad C_2 = [C,B] = \partial_z -\partial_q. \end{equation}
Hence the operator is \textit{hypoelliptic}~\cite{Hormander1967}. Furthermore,
\begin{align}\label{[A,A] [A,C] [A,C2] [A,A*] [C,A*] [C2,A*] [C2,B]}
&[A,A]=0\quad [A,C]=0\quad [A,C_2]=0,   \nonumber\\
&[A,A^*]=Id\qquad [C,A^*]=0\qquad [C_2,A^*]=-Id, \nonumber\\
&[C_2,B]=-\partial^2V\partial_p-\partial_p, \qquad\nonumber\\
&[C,C^*]=Id\qquad[C_2^*,C_2]=-Id-\partial_q^2V,
\end{align}
where  $Id$ is the identity operator.
\subsection{Proof of Theorem \ref{thm:converg_equil} }
\begin{proof}
We will use Theorem \ref{thm:hypocoercivity} . To this end, set
\begin{equation*}
P=A^*A+C^*C+C_2^*C_2
\end{equation*}
and notice that $Ker(P)=\mathcal{K}=:Ker\op$ contains only constants; in fact 
$$
Ker(P)=Ker(A^*A)\cap
Ker(C^*C)\cap Ker(C_2^*C_2)=Ker(A)\cap Ker(C)\cap Ker(C_2).
$$ 
To show that  $\mathcal{K}=Ker(A^*A)\cap Ker(C^*C)\cap Ker(C_2^*C_2)$: the inclusion $\supseteq$ is obvious. For the other inclusion: if $h\in \mathcal{K}$ then $\| Ah\|^2+\|Ch\|^2+\| C_2 h\|^2=0 \Rightarrow Ah=Ch=C_2h=0$.\\
The above mentioned  theorem requires two sets of hypotheses to be fulfilled. Hypothesis 1,2 and 3. in Theorem \ref{thm:hypocoercivity} are quantitative assumptions, which are  satisfied in our case with $N=2$, $C_0=A$, $C_1=C$, $R_1=R_2=0$, $R_3=[C_2,B]$ (this is  to have $C_3=0$) and thanks to Assumption \ref{assump:potential}(iii).
Hypothesis 4. requires, in our case, for the operator $P$ to be $\kappa$-\textit{coercive} on 
$\mathcal{K}^{\perp}\cong \fspace / \mathcal{K}$.
The coercivity of $P$ is equivalent to 
$$
\|A h\|^2+\|C h\|^2+\|C_2 h\|^2\geq\kappa\|h\|^2,
$$ 
that is, more explicitly,
$$
\|\nabla_z h\|^2+ \|\nabla_p h\|^2+ 
\|\left(\nabla_z-\nabla_q\right)h\|^2 \geq \kappa \|h\|^2.
$$
Using the fact that $\|a-b\|^2\geq \frac{\|a\|^2}{3}-\frac{\|b\|^2}{2}$, we have
$$
\|\nabla_z h\|^2+ \|\nabla_p h\|^2+ 
\|\left(\nabla_z -\nabla_q\right)h\|^2 \geq\frac{1}{3} \left(  \|\nabla_z h\|^2+ \|\nabla_p h\|^2+\|\nabla_q h\|^2\right)
$$
so we just need
$$
\|\nabla_z h\|^2+ \|\nabla_p h\|^2+\|\nabla_q h\|^2 \geq\kappa \|h\|^2
$$
to hold true. Since  $\mu_{\beta}$ is a product measure, we only need to verify that
$$
\int \vert \nabla_q h\vert^2 e^{-V(q)} dq\geq \mu\int (h-\langle h\rangle)^2 e^{-V(q)} dq
$$   
holds true for some constant $\mu$, where the notation $\langle h\rangle :=\int he^{-V(q)dq}$ has been used. It is a standard result that if $V(q)\in C^2(\R^d)$ is such that $e^{-V(q)}/\mathcal{Z}$ is a probability density  and 
\begin{equation}\label{potential_Poinc}
\frac{\mid \nabla V(q)\mid^2}{2}-\Delta V(q)
\stackrel{\mid q\mid\rightarrow\infty}{\longrightarrow}+\infty
\end{equation}
then $e^{-V(q)}/\mathcal{Z}$ satisfies a Poincar\'e inequality (see, e.g.,~\cite[Thm. A.1]{Vil04HPI} ). From Assumption \ref{assump:potential}(iii), Condition \eqref{potential_Poinc} is satisfied.
We can conclude that there exist a scalar product $((\cdot,\cdot))$ inducing a norm equivalent to the inhomogeneous norm of $H^1$ and a constant $\hat{\lambda}>0$ such that $\op$ is coercive in this norm:
$$
\forall h\in \fspace/ \mathcal{K}, \qquad((h,\op h))\geq\hat{\lambda}((h,h)).
$$     
This implies that $\op$ is hypocoercive in this norm, hence it is hypocoercive on $\fspace/ \mathcal{K}$ endowed with the $\|\cdot\|_{H^1}$ norm:
\begin{equation}\label{exponential convergence to equil datum in H^1}
\|e^{-t\op\p} h_0\|_{H^1}\leq C e^{-\lambda t}\|h_0\|_{H^1.}
\end{equation}
\end{proof} 
\begin{remark}
The orthogonal space to  $\mathcal{K}$ is the same with respect to both the $(\cdot,\cdot)_1$ and the $(\cdot,\cdot)_{H^1}$ norms; moreover, since $P$ is coercive, these two norms are equivalent.   
\end{remark}

\begin{remark}
Theorem \ref{thm:hypoeelipticity short time asymp} in Appendix A allows us to state a similar result when the initial datum is in $\fspace$. In fact
\begin{equation}\label{regularity estimate}
\|e^{-t\op}h\|_{H^1}\leq \frac{c}{t^{\frac{5}{2}}}\|h\|,\quad t\in(0,1].
\end{equation}
So, putting together (\ref{exponential convergence to equil datum in H^1}) and (\ref{regularity estimate}) we get, for $0<t_0<t$, $t_0<1$:
\begin{align}\label {exponential conv equil datum in L^2}
\|e^{-t\op}h_0\|_{H^1}=&\|e^{-(t-t_0)\op}e^{-t_0\op}h_0\|_{H^1}=
\| e^{-(t-t_0)\op}h_{t_0}\|_{H^1}\nonumber\\
&\leq c e^{-\lambda (t-t_0)}\|h_{t_0}\|_{H^1}\leq c e^{-\lambda(t-t_{0})} 
\|e^{- t_0\op} h_0 \|_{H^1}\nonumber\\
&\leq c \frac{ e^{-\lambda (t-t_0)}}{t_0^{\frac{5}{2}}}\|h_0\|.
\qquad\qquad\qquad \nonumber
\end{align}
\end{remark}
\begin{remark}\label{rem:m,d>1}
The proof is identical when $m,d>1$. In this case we can think of $A$ as a matrix of operators, see (\ref{l=b+A*A general}).
\end{remark}
%
%%%%%%%%%%%%%%%%%%%%%%%%%%%%%%%%%%%%%%%%%%%%%%%%%%%%%%%%%%%%%%%%%%
%
\subsection{Proof of Theorem \ref{thm:converg_rel_entropy}}
\begin{proof}Let $f_t$ denote the law of the process $\bx$. 
We set $f_t=\rho h_t$. Then $h_t$ satisfies the equation
\begin{equation}\label{fake kolmog}
\partial_t h_t = B h_t-A^{\ast} A h_t\,\mbox{.}
\end{equation}
We  apply Theorem \ref{thm:long time entropic} to the operator $\mathcal{F}=-B+A^{\ast}A$ with 
$$
A=-\partial_z,\, C_1=-\partial_p,\, C_2=-\partial_q,\, 
Z_2=Id,\, R_2=-\partial_z.
$$
Furthermore Asuumption \ref{assump:potential} (i) and (iii) together with the Holley-Strook perturbation Lemma imply that $\mathcal{Z}^{-1}e^{-V(q)}$ satisfies a Logarithmic Sobolev Inequality (LSI).   

%(for example we can put ourselves in the hypothesis $V(q)=$uniformly convex+ %bounded)
Hypotheses 1,2 and 4 are automatically satisfied. We put $C_2=\partial_q$ and we added the remainder $R_2$ in order to fulfill hypothesis 4. Hypothesis 3. is satisfied on account of Assumption \ref{assump:potential} (iii).
Now consider the relative entropy  $H_{\rho}(f)$,
\begin{equation}\label{def of the Boltzmann H functional}
H_{\rho}(f)=\int f \log\left( \frac{f}{\rho}\right)dq \,dp\, dr=\int h \log h \,d\rho\qquad f= \rho h
\end{equation}  
and the Fisher information $I_{\rho}(f)$
\begin{equation}\label{def of the Fisher information}
I_{\rho}(f)=\int f\vert\nabla \log (h)\vert^2 dqdpdr=\int h \vert\nabla\log h \vert^2\,d\rho\qquad f= \rho h.
\end{equation}
Then if the initial datum has finite relative entropy, we obtain that
\begin{equation}\label{convergence in entropy for f}
H_{\rho}(f_t)=\mathcal{O}(e^{-t\alpha})
\end{equation}
for some $\alpha>0$ and for $t>0$.
If the initial datum has also finite Fisher information then 
\begin{equation}\label{convergence of the fisher information}
I_{\rho}(f_t)=\mathcal{O}(e^{-t\alpha}),
\end{equation}
as well. 
\end{proof}
\begin{remark}
We remark that (\ref{convergence of the fisher information}), together with the LSI, implies (\ref{convergence in entropy for f}).
\end{remark}
\begin{remark}
In viev of the LSI, it is interesting to notice that, applying Theorem \ref{thm:short time entropic}, we get the following bounds 
\begin{equation}\label{regularization in LlogL context}
\int h(t)\vert C_k \log h(t)\vert^2 d\rho\leq\frac{c}{t^{2k+1}}
\int h_0 \log h_0 \,d\rho,
\end{equation} 
for $k=0,1,2$ and  $c$ an explicitly computable positive constant.
\end{remark}

%%%%%%%%%%%%%%%%%%%%%%%%%%%%%%%%%%%%%%%%%%%%%%%%%%%%%%%%%%%%%%%%%%%%%%%
%
\section{Bounds on the derivatives of the Markov semigroup}
\label{sec:derivatives}

Throughout this section we well be using the notation $u=e^{-t\op}u_0$. We  introduce  the Lyapunov function
\begin{eqnarray}\label{STA Lyap funct}
F(t) &=& a_0t\|Au\|^2+a_1t^3\|Cu\|^2+a_2t^5\|C_2u\|^2\nonumber \\ &&
+b_0t^2(Au,C u)+t^4b_1(C u,C_2u)+b_2\|u\|^2, \qquad t\in(0,1],
\end{eqnarray}
where $a_j, b_j, j=0,1,2$ are positive constants to be chosen.
\begin{lemma}
There exist constants $a_j, b_j, j=0,1,2$ such that the time derivative $\partial_t F$ of the Lyapunov function along the semigroup is negative.
\end{lemma}
\begin{proof}
We will calculate the time derivative of each term in (\ref{STA Lyap funct}) separately  and  using the explicit relations (\ref{[A,A] [A,C] [A,C2] [A,A*] [C,A*] [C2,A*] [C2,B]}):
\begin{eqnarray*}
\partial_t \|u\|^2    & = & -2(\op u,u)=-2\|Au\|^2, \\
\partial_t(Au,Au)     & = & -2(Cu,Au)-2\|A^{\ast}Au\|^2=-2(Cu,Au)-2\|Au\|^2-2\|A^2u\|^2,\\
\partial_t(Cu,Cu)     & = &  -2\|ACu\|^2-2(C_2u,Cu),\\
\partial_t(C_2u,C_2u) & = & ((2+\partial_q^2V)C_2u,Cu)-2\|AC_2u\|^2+2(Au,C_2u),\\
\partial_t(Au,Cu)     & = &  -2(A^2u,ACu)-(Au,Cu)-\|Cu\|^2-(Au,C_2u),\\
\partial_t(Cu,C_2u)   & = &   -\|C_2u\|^2-2(ACu,AC_2u)+2\|Cu\|^2+(Cu,Au).
\end{eqnarray*}
Putting everything together we obtain
\begin{eqnarray}\label{Derivative of the Lyapunov function}
\partial_t F(t)&=&-2a_0t\|A^2u\|^2-2a_1t^3\|ACu\|^2-2a_2t^5\|AC_2u\|^2\nonumber\\
& &-2b_0t^2(A^2u,ACu)-2b_1t^4(ACu,AC_2u)\nonumber\\
& &+(-2a_0t+a_0-2b_2)\|Au\|^2+(3a_1t^2+2b_1t^4-b_0t^2)\|Cu\|^2\nonumber\\
& &+(5a_2t^4-b_1t^4)\|C_2u\|^2+(2b_0t-2a_0t-b_0t^2+b_1t^4)(Au,Cu)\nonumber\\
& &+(4b_1t^3-2a_1t^3+2a_2t^5)(Cu,C_2u)+(2a_2t^5-b_0t^2)(Au,C_2u).\nonumber
\end{eqnarray}
Now we estimate the sum of the first and of the second line (i.e. the sum of all the terms where $A^2$, $AC$ and $AC_2$ appear). For $t\in(0,1]$ we have
$$
(\ref{Derivative of the Lyapunov function})_1+(\ref{Derivative of the Lyapunov function})_2\leq-2a_0t \|A^2u\|^2+2b_0t^2\|A^2u\| \|ACu\|
$$
$$
+2b_1t^4 \|ACu\|\|AC_2u\|-2a_1t^3\|ACu\|^2-2a_2t^5\|AC_2u\|^2
$$
$$
\leq-2a_0t\|A^2u\|^2+b_0^2t\|A^2u\|^2+t^3\|ACu\|^2-2a_1t^3\|ACu\|^2
$$
$$
+b_1^2t^3\|ACu\|^2+t^5\|AC_2u\|^2-2a_2t^5\|AC_2u\|^2.
$$
Similarly for  the sum of the remaining terms (those with $A$, $C$ and $C_2$)we have
$$
(\ref{Derivative of the Lyapunov function})_3+(\ref{Derivative of the Lyapunov function})_4+(\ref{Derivative of the Lyapunov function})_5\leq(-2a_0t+a_0-2b_2)\|Au\|^2
$$
$$
+(2b_0t+2a_0t+b_0t^2+b_1t^4)\|Au\|\|Cu\|+(2a_2t^5+b_0t^2)\|Au\|\|C_2u\|
$$
$$
+(3a_1t^2+2b_1t^4-b_0t^2)\|Cu\|^2+(5a_2t^4-b_1t^4)\|C_2u\|^2
$$
$$
+(4b_1t^3+2a_1t^3+2a_2t^5)\|Cu\|\|C_2u\|
$$
$$
\leq(-2a_0t+a_0-2b_2)\|Au\|^2+a_0^2\|Au\|^2+\|Cu\|^2
$$
$$
+\frac{3}{2}b_0^2\|Au\|^2+\frac{3}{2}t^2\|Cu\|^2+\frac{1}{2}b_1^2\|Au\|^2+\frac{t^4}{2}\|Cu\|^2
$$
$$
+a_2^2t^5\|Au\|^2+t^5\|C_2u\|^2+\frac{t^2}{2}b_0^2\|Au\|^2+\frac{t^2}{2}\|C_2u\|^2
$$
$$
+2b_1^2t^3\|Cu\|^2+t^3\|C_2u\|^2+a_1^2t^3\|Cu\|^2+t^3\|C_2u\|^2
$$
$$
+a_2^2t^5\|Cu\|^2+t^5\|C_2u\|^2.
$$
Choosing the constants in such a way that
\fbox{$b_2\gg a_0\gg b_0\gg a_1\gg b_1\gg a_2>1/c$}, where $c$ is a constant depending on the bound on the second derivative of the potential, 
we obtain that $\partial_tF<0$ $\forall t\in(0,1]$. 
\end{proof}
\begin{proof}[Proof of Theorem \ref{thm:short_time_asymp}]
We use the previous Lemma to deduce
\begin{eqnarray*}
& &a_0t\|Au\|^2+a_1t^3\|Cu\|^2+a_2t^5\|C_2u\|^2\\
& &+b_0t^2(Au,Cu)+t^4b_1(Cu,C_2u)+b_2\|u\|^2<b_2\|u_0\|^2.
\end{eqnarray*}
This, in turn, implies that
\begin{eqnarray}\label{short time behavior}
& &\|\nabla_z u\|^2=\|Au\|^2<\frac{\kappa}{t}\|u_0\|^2,\nonumber\\
& &\|\nabla_p u\|^2=\|Cu\|^2<\frac{\kappa}{t^3}\|u_0\|^2,\nonumber\\
& &\frac{\|\nabla_q u\|^2}{3}-\frac{\|\nabla_zu\|^2}{2}\leq\|\nabla_q u-\nabla_z u\|^2=\|C_2u\|^2<\frac{\kappa}{t^5}\|u_0\|^2\nonumber\\
& &\Rightarrow \|\nabla_q u\|^2\leq \frac{\kappa}{t^5}\|u_0\|^2,\nonumber
\end{eqnarray} 
where $\kappa$ is an explicitly computable positive constant. The previous inequalities are justified by the fact that
$$
a_0t\|Au\|^2+a_1t^3\|Cu\|^2+a_2t^5\|C_2u\|^2
+b_0t^2(Au,Cu)+t^4b_1(Cu,C_2u)
$$
$$
\geq(a_0t-\frac{b_0^2}{2}t^2)\|Au\|^2+(a_1t^3-\frac{t^2}{2}-t^4\frac{b_1^2}{2})\|Cu\|^2+(a_2t^5-\frac{t^4}{2})\|C_2u\|^2
$$
and the second line is positive thanks to the choice of the constants we made.
\end{proof}
Remark \ref{rem:m,d>1} holds also in this case.
\begin{remark}
From the estimates (\ref{e:short_time}), similar estimates on $A^{\star}e^{-t\op^{\bullet}}$, $e^{-t\op^{\star}}A^{\bullet}$, $C^{\star}e^{-t\op^{\bullet}}$, $e^{-t\op^{\star}}C^{\bullet}$, $C_2^{\star}e^{-t\op^{\bullet}}$ and $e^{-t\op^{\star}}
C_2^{\bullet}$ follow, where $\star$ and $\bullet$ stand for either the $\fspace$-adjoint or nothing. In fact:
\begin{description}
\item(i) $(Ae^{-t\op}f,g)=(f,e^{-t\op^*}A^*g)\leq\|Ae^{-t\op} f\|\|g\|\leq \frac{\kappa}{\sqrt{t}}\|f\|\|g\|$\\
$\Rightarrow (f,e^{-t\op^*}A^*g)\leq \frac{\kappa}{\sqrt{t}}\|f\|\|g\| $, choose $f=e^{-t\op^*}A^*g$ and the result on $e^{-t\op^*}A^*$ follows.
\item(ii)Using $[A,A^*]=Id$   we have 
$\|A^*e^{-t\op}u_0\|^2=\|A^*u\|^2=\|Au\|^2+\|u\|^2$, hence the estimate for $A^*e^{-t\op}$. Taking the adjoint as in (i) we get the result for $e^{-t\op^*}A$.
\item(iii) For $Ae^{-t\op^*}$ we can just repeat the proof we wrote for $Ae^{-t\op}$, since the only thing that changes when considering $\op^*$ is the sign of $B$, which doesn't play any role in the proof. Now, by acting as in (i) and (ii), we obtain the results for $e^{-t\op}A^*$,$A^*e^{-t\op^*}$ and $e^{-t\op}A$. 
\end{description} 
\end{remark}
%
%%%%%%%%%%%%%%%%%%%%%%%%%%%%%%%%%%%%%%%%%%%%%%%%%%%%%%%%%%%%%%%%%%%%%%%
%
\section{The Homogenization Theorem}
\label{sec:homogen}

In this section we prove Theorem~\ref{thm:homogenization}. The proof of this theorem is based on standard techniques, namely the central limit theorem for additive functionals of Markov processes~\cite{kipnis, Land03, papan_varadhan}, which in turn is based on the martingale central limit theorem~\cite[Thm. 7.1.4]{EthKur86}. In order to apply these techniques we need to study the Poisson equation  
\begin{equation}\label{e:poisson}
\op u = f.
\end{equation}
The boundary conditions for~\eqref{e:poisson} are that $ u \in \fspace$ and it is periodic in $q$.

\begin{prop}\label{prop:poisson} 
Let $f\in\fspace \cap C^{\infty}(X)$ with $\int_X f \mu_{\beta} (d {\bf x})$. Then the Poisson equation~\eqref{e:poisson} has a unique smooth mean zero solution $u \in \fspace \cap C^{\infty}(X)$.
\end{prop}

The proof of Theorem~\ref{thm:homogenization} follows now from the above proposition.

\begin{proof}[Proof of Theorem \ref{thm:homogenization}]
To simplify the notation we present the proof for $d=1$.
When $d>1$ the same proof applies to the one-dimensional projections $q^e=q \cdot e$. In this case the diffusion coefficient $D$ is replaced by the projections of the diffusion tensor $D^e=De\cdot e$.\\
We consider the process $\bx$ on $X$ with stationary initial conditions. For non-stationary initial conditions we need to combine the analysis presented below with the exponential convergence to equilibrium, Theorem~\ref{thm:converg_equil}. %Let $e$ be an arbitrary unit vector in $\R^d$. We will use the notation %$\phi^e :=\phi \cdot e, \, p^e:= p \cdot e$ etc. We consider the Poisson %equation
%\begin{equation}\label{e:poisson_phi}
%-L \phi^e = p^e.
%\end{equation}
Since $p \in \fspace \cap C^{\infty}(X)$ and centered with respect to the invariant measure $\mu_{\beta}(d {\bf x})$, Proposition~\ref{prop:poisson} applies and there exists a unique mean zero solution $\phi\in \fspace \cap C^{\infty}(X)$ to the problem
\begin{equation}\label{poisson phi}
\op \phi=p.
\end{equation} 
We use It\^{o}'s formula to obtain
$$
d \phi = \op \phi \, dt + \sum_{j=1}^m \sqrt{2 \alpha_j \beta^{-1}} \partial_{z_j} \phi \, dW_j.
$$
We combine this, together with (\ref{poisson phi}) and the equations of motion to deduce
\begin{eqnarray*}
q_{\eps}(t) &:= & \eps q (t/\eps^2) \\
              & = & \eps q (0) + \eps \int_0^{t/\eps^2} p(s) \, ds \\
              & = & \eps q (0) - \eps \big( \phi (q (t/\eps^2), p(t/ \eps^2), z(t/ \eps^2) - \phi (q (0), p(0), z(0)  \big) \\ && + \eps \sum_{j=1}^m  \int_0^{t/\eps^2} \sqrt{2 \alpha_j \beta^{-1}} \partial_{z_j} \phi  \, d W_j (s) \\
              & =: & \eps R^{\eps} + M^\eps.  
\end{eqnarray*} 
Our stationarity assumption, together with the fact that $\phi \in \fspace$, imply that
$$
\|R^\eps \| \leq C.
$$
To study the martingale term $M^\eps$ we use the martingale central limit theorem~\cite[Thm. 7.1.4]{EthKur86} or ~\cite[Thm 3.33]{PavlSt08}. We have that $M^\eps(0) = 0$, that $M^\eps (t)$ has continuous sample paths and, by stationarity, that it has stationary increments. Furthermore, by stationarity and the fact that the Brownian motions $W_i(t), \, i=1, \dots m$ are independent, we deduce that 
\begin{equation*}
\lim_{\eps \rightarrow 0} \langle M^{\eps}_t \rangle  =2  \sum_{i=1}^m \alpha_i\beta^{-1}\|\partial_{z_i} \phi\|^2 t  \quad \mbox{in} \; \; L^1_{\rho}.
\end{equation*}
The above calculations imply that the rescaled one dimensional projection $q_{\eps}(t) := \eps q(t/\eps^2) $ converges weakly in $C([0,t]; \R)$ to a Brownian motion $\sqrt{ 2 D} W(t)$ where 
\begin{equation}\label{e:deff_e}
D = 2  \sum_{i=1}^m \alpha_i\beta^{-1}\|\partial_{z_i} \phi\|^2 .
\end{equation}
\begin{remark}
Notice that when $d>1$ the convergence of the one dimensional projections does not imply the convergence of the process $q_{\eps}(t) = \eps q(t/\eps^2)$. The proof of this result, which is also based on the analysis of the Poisson equation, is very similar and it is omitted.
\end{remark}
To prove estimate~\eqref{e:deff_estim}, we first show the upper bound and than the fact that the diffusion coefficient is bounded away from zero (when we consider periodic solutions only). We set $\phi = g_i + \frac{1}{\lambda_i} z_i$ and use the Poisson equation (\ref{poisson phi}) to obtain
$$
\op g_i=-\frac{\alpha_i}{\lambda_i}z_i,
$$
from which we obtain the estimate
\begin{eqnarray*}
\alpha_i\beta^{-1} \|\partial_{z_i}g_{i}\|^2
& \leq &
\sum_{j=1}^m\alpha_j\beta^{-1} \|\partial_{z_j}g_{i}\|^2=(\op g_i, g_i)
\\ & = &
\frac{\alpha_i}{\beta\lambda_i} \int g_i\partial_{z_i}\rho\, d {\bf x}
=-\frac{\alpha_i}{\beta\lambda_i} \int\rho\,\partial_{z_i}g_i d {\bf x}
\\ & \leq &
\frac{\alpha_i}{\beta\lambda_i} \|\partial_{z_i}g_i\|.
\end{eqnarray*}
Consequently $\|\partial_{z_i}g_i\|\leq\frac{1}{\lambda_i}$.
From this we obtain the following estimate on the diffusion coefficient $D$
\begin{eqnarray*}
D & = &  \sum_{i=1}^m \alpha_i\beta^{-1}\|\partial_{z_i}\phi\|^2 = \frac{1}{\beta}\sum_{i=1}^m \alpha_i\|\partial_{z_i}g_i+\frac{1}{\lambda_i}\|^2
\\ & \leq &
\frac{2}{\beta}\sum_{i=1}^m \alpha_i\left( \|\partial_{z_i}g_i\|^2+\frac{1}{\lambda_i^2}\right)
\\ & \leq &
\frac{4}{\beta}\sum_{i=1}^m\frac{\alpha_i}{\lambda_i^2}.
\end{eqnarray*}
The fact that $D>0$ is easily seen by contradiction. If $D=0$ then, by (\ref{e:deff_e}), $\|\partial_{z_i}\phi\|^2=0 \; \forall i$. Hence $\phi=\phi(q,p)$ and 
$$
\op \phi=-p\partial_q\phi+\partial_qV\partial_p\phi+\sum_{i=1}^m \lambda_iz_i\partial_p\phi=p.
$$
Multiplying both sides by $e^{z_i^2/2}$ and then integrating with respect to $z_i$ we get
\begin{eqnarray*}
&&-\int p\partial_q\phi\, e^{z_i^2/2}dz_i+\int\partial_qV\partial_p\phi\,e^{z_i^2/2}dz_i
\\ &&
+\int \lambda_i z_i^2e^{z_i^2/2}dz_i+\sum_{j\neq i} \int\lambda_i z_i z_j\partial_p\phi e^{z_i^2/2}dz_i
\\ &=&
\int p z\, e^{z_i^2/2}dz_i,
\end{eqnarray*}
from which we conclude that $\lambda_i\partial_p\phi =0$ for all $ i$. Hence $\phi=\phi(q)$.
By the same reasoning we get that $-p \partial_q \phi=p$, which does not have a periodic solution. 
\end{proof}

\noindent {\it Proof of Proposition~\ref{prop:poisson}.}

The scheme of the proof, which is similar to the proof of~\cite[Lemma 2.1]{papan_varadhan}, is as follows. We consider the Poisson equation $\cL \phi =f$ where $f \in L_{\rho} \cap C^{\infty}(X)$ and centered with respect to the invariant measure $\mu_{\beta} (d {\bf x})$.
\begin{enumerate}
\item 
We consider the modified problem 
\begin{equation}\label{regularized problem}
\lambda u_{\nu}+\op u_{\nu}+\nu \hat{P}u_{\nu}=f \qquad \lambda,\nu>0
\end{equation}
where $\hat{P}$ is a regularizing operator (i.e. the operator $\op + \hat{P}$ is uniformly elliptic) so that  the Lax-Milgram theorem applies; hence  the weak solution to (\ref{regularized problem}) is unique $\forall \nu>0$ and letting $\nu\rightarrow 0$, by the uniqueness of the weak limit, we get existence and uniqueness of the solution to
\begin{equation}\label{resolvent equation}
\lambda u+\op u= f \qquad  \lambda>0
\end{equation}
(notice that $-\op$ is the generator of a Markov semigroup so by Hille-Yosida theorem the set $\left\{\lambda\in R: \lambda>0 \right\}$ is contained in the resolvent of the operator $-\op$)
\item
Set $\op_{\lambda}u=\lambda u+\op u=\lambda u +f$  (the last equality has to hold in distribution if we want $u$ to satisfy (\ref{e:poisson}))
$$
\Rightarrow u=\op_{\lambda}^{-1}(\lambda u+f)=\lambda \op_{\lambda}^{-1}u
+\op_{\lambda}^{-1} f
$$
and defining  $\op_{\lambda}^{-1} f=h$ we get  $(\frac{1}{\lambda}Id-\op_{\lambda}^{-1})u=\tilde{h}$, $\tilde{h}=h/\lambda$.
\item Once proven that $\op_{\lambda}^{-1}$ is compact we are done; in fact in this case the Fredholm Theorem applies so either the solution to  
$(\frac{1}{\lambda}Id-\op_{\lambda}^{-1})u=\tilde{h}$ exists and is unique (and hence, by construction the solution to (\ref{e:poisson}) is unique) or $(\frac{1}{\lambda}Id-\op_{\lambda}^{-1})u=0$ admits a nonzero solution. We can rule out the latter option because $(\frac{1}{\lambda}Id-\op_{\lambda}^{-1})u=0 \Leftrightarrow \op u=0$; since we know that $Ker{\op}$ contains only constants and we are asking for $u$ to have mean zero we can conclude that $\op u=0 \Leftrightarrow u=0$ and we are done.
\end{enumerate} 
Now the details. Choose $\hat{P}=C^{\ast}C+G^{\ast}G$ where $G=\nabla_q$ and notice that now $B=G^{\ast}C-C^{\ast}G+A^*C-C^*A$. Let us check that the hypothesis of Lax-Milgram Theorem hold: 
$$
\lambda (u_{\nu},u_{\nu})+\|Au_{\nu}\|^2+\nu\|Cu_{\nu}\|^2+\nu\|Gu_{\nu}\|^2
$$
$$
\geq\lambda(u_{\nu},u_{\nu})+\nu\|\nabla_{qpr}u_{\nu}\|^2\geq min\{\lambda,\nu\}\|u_{\nu}\|_{H^1}
$$
and 
$$
\lambda(u_{\nu},v)+(Bu_{\nu},v)+(Au_{\nu},Av)+\nu(Cu_{\nu},Cv)+\nu(Gu_{\nu},Gv)
$$
$$
\leq\lambda(u_{\nu},v)+\|Cu_{\nu}\|\|Gv\|+\|Gu_{\nu}\| \|Cv\|+\|Au_{\nu}\|\|Av\|
$$
$$
+\nu \|Cu_{\nu}\|\|Cv\|+\nu\|Gu_{\nu}\|\|Gv\|+\|Cu_{\nu}\|\|Av\|+\|Au_{\nu}\|\|Cv\|
$$
$$
\leq \alpha \|u_{\nu}\|_{H^1}\|v\|_{H^1},
$$
with $\alpha$ a positive constant depending on $\nu$ and $\lambda$.
Hence the weak solution to (\ref{regularized problem}) exists and is unique. Moreover, from (\ref{regularized problem}) we also have that
$$
\lambda \|u_{\nu}\|^2+\|Au_{\nu}\|^2+\nu \left( \|Cu_{\nu}\|^2+\|Gu_{\nu}\|^2\right)=(f,u_{\nu})
\leq\|f\|\|u_{\nu}\|,
$$
hence $\lambda \|u_{\nu}\|\leq c$, $\|Au_{\nu}\|^2\leq c$ and $\nu  \|Cu_{\nu}\|^2, \nu\|Gu_{\nu}\|^2 \leq c$. Now taking $ u_{\nu}/ \nu $ as test function we get that $\|u_{\nu}\|_{H^1}\leq c \|f\|$ and  we can let $\nu\rightarrow 0$ obtaining that $u_{\nu}$ converges weakly to the solution of (\ref{resolvent equation}). Recalling that $u=\op_{\lambda}^{-1}f$ and that $u$ the weak limit of $\left\{u_{\nu}\right\}$, we can write
$$
\|\op_{\lambda}^{-1}f\|_{H^1}=\|u\|_{H^1}\leq \liminf_{\nu\rightarrow 0}\|u_{\nu}\|_{H^1}\leq c\|f\|
$$
and using the fact that $H^1_{\rho}$ is compactly embedded in $L^2_{\rho}$ we can conclude that the resolvent is compact. This completes the proof.
\qed

%
%
%%%%%%%%%%%%%%%%%%%%%%%%%%%%%%%%%%%%%%%%%%%%%%%%%%%%%%%%%%%%%%%%%%%%%%%%
%
%
\section{The White Noise Limit}
\label{sec:markovian}
Throughout this section $C$ denotes a generic constant, which is independent of $\epsilon$. To simplify the notation we present the proof in one dimension. The proof is exactly the same in arbitrary dimensions. Let  $(Q(t),P(t))\in \mathbb{T}\times \mathbb{R} $ be the solution to the  system (\ref{limit system for white noise scaling}). Then
$$
 \mid q(t)-Q(t)\mid \leq \mid q(0)-Q(0)\mid + \int_0^t   \mid p(s)-P(s)\mid ds. 
$$
Setting $\theta_i=\lambda_i^2/\alpha_i$ we get that
$$
\dot{p}(t)-\dot{P}(t)=-\partial_qV(q)+\partial_qV(Q)+\sum_{i=1}^m\theta_i(P(t)-p(t))
-\sqrt{\epsilon}\sum_{i=1}^m\frac{\lambda_i}{\alpha_i} \dot{z}_i(t).
$$
Hence, for any $r>2$,
\begin{eqnarray*}
\eta(T) &:=&  E\sup_{t\in[0,T]}\left\{ \mid q(t)-Q(t)\mid^r+\mid p(t)-P(t)\mid^r\right\}
\\ &\leq& 
C T^{r-1}\int_0^T E\sup_{s\in[0,t]}\mid q(s)-Q(s)\mid^r \, dt
\\ &&+
C \left(\sum_{i=1}^m \theta_i^r\right) T^{r-1}\int_0^T E\sup_{s\in[0,t]}\mid p(s)-P(s)\mid^r \, dt
\\ &&+
C \epsilon^{\frac{r}{2}}\sum_{i=1}^m \left(\frac{\lambda_i}{\alpha_i}\right)^r
E\sup_{t\in[0,T]}\mid z_i(t)-z_{i}(0)\mid^r+E\mid p(0)-P(0)\mid^r+E\mid q(0)-Q(0)\mid^r.
\end{eqnarray*}
From Gronwall's Lemma we get
$$
\eta(T)\leq C\left[\sum_{i=1}^m \left(\frac{\lambda_i}{\alpha_i}\right)^r
E\sup_{t\in[0,T]}\mid z_i(t)-z_i(0)\mid^r\right] \epsilon^{\frac{r}{2}}e^{CT}
$$
and the result now follows from Proposition \ref{prop:stima mark limit}. 

\begin{prop}\label{prop:stima mark limit}
With the same notation and assumptions of Theorem 
\ref{thm:markovian limit} the following estimate holds true:
\begin{eqnarray*}
\xi(T)
&:=&
\sum_{i=1}^m E\sup_{t\in[0,T]}\mid z_i(t)\mid^r + E\sup_{t\in[0,T]}\mid p(t)\mid^r\leq C\left(1+T+\frac{\alpha}{\epsilon}T+\sqrt{\frac{\alpha}{\epsilon}T}\right)
\\ &&+
C\left(1+\frac{\lambda}{\sqrt{\epsilon}}+\frac{\alpha+\sqrt{\alpha}}{\epsilon} \right)\int_0^T\xi(t)dt +E\mid p(0)\mid^r+\sum_{i=1}^m 
E\mid z_i(0)\mid^r.
\end{eqnarray*}
In particular
$$
\sum_{i=1}^m E\sup_{t\in[0,T]}\mid z_i(t)\mid^r\leq C\epsilon^{-1}(T+\sqrt{T})
e^{CT\epsilon^{\frac{r}{2}-1}}
$$
hence
$$
\epsilon^{\frac{r}{2}}\sum_{i=1}^m  \left(\frac{\lambda_i}{\alpha_i}\right)^rE\sup_{t\in[0,T]}\mid z_i(t)-z_i(0)\mid^r\longrightarrow 0 \qquad\forall r>2.
$$
\end{prop}

\begin{proof}
We use  It\^o's formula to deduce that
\begin{eqnarray*}
 \frac{1}{r} p(t)^r &=& \frac{1}{r} p(0)^r-\int_0^t (\partial_q V (q(s)) p(s)^{r-1}) \, ds+\frac{1}{\sqrt{\epsilon}}
\sum_{i=1}^m\int_0^t \lambda_i p(s)^{r-1}z_i(s) \, ds \nonumber\\
\frac{1}{r} z_i(t)^r & = &  \frac{1}{r} z_i(0)^r-\frac{\alpha}{\epsilon}\int_0^t  z_i(s)^r \, ds-\frac{\lambda}{\sqrt{\epsilon}} 
\int_0^t( p(s) z_i(s)^{r-1})ds
\\ &&+
\frac{\alpha}{\epsilon}\int_0^t z_i(s)^{r-2}ds +\sqrt{\frac{2\alpha_i}{\epsilon}}
\int_0^t z_i(s)^{r-1} \, dW_i(s).\nonumber
\end{eqnarray*}
Setting $\lambda=\max\left\{\lambda_i, i=1...m\right\}$, using the boundedness of $\partial_qV$ on the torus and  H\"older's inequality, we have
\begin{align*}
E\sup_{t\in[0,T]}\mid p(t)\mid^r \leq& CT+ 
C \left( 1+\frac{1}{\sqrt{\epsilon}}\lambda \right) \int_0^T E \sup_{s\in[0,t]}\mid p(s)\mid^r \, dt\\
&+\frac{1}{\sqrt{\epsilon}}m\lambda\int_0^T \sum_{i=1}^m E \sup_{s\in[0,t]}\mid z_i(s)\mid^r dt+E\mid p(0)\mid^r.
\end{align*}
Similarly,
\begin{eqnarray*}
\sum_{i=1}^m E\sup_{t\in[0,T]}\mid z_i(t)\mid^r 
&\leq& 
C \frac{1}{\sqrt{\epsilon}}\lambda\int_0^T
 E \sup_{s\in[0,t]}\mid p(s)\mid^r dt 
\\ &&
+C(\frac{\lambda}{\sqrt{\epsilon}}+2\frac{\alpha}{\epsilon}+\sqrt{\frac{2\alpha}{\epsilon}})\int_0^T \sum_{i=1}^m E \sup_{s\in[0,t]}\mid z_i(s)\mid^r dt
\\ &&+
m\frac{\alpha}{\epsilon}T+Cm\sqrt{\frac{2T\alpha}{\epsilon}}+\sum_{i=1}^mE\mid z_i(0)\mid^r.
\end{eqnarray*}
In the above we used the following estimate, which is a consequence of the Burkholder-Davies-Gundy inequality 
\begin{align}
E \sup_{t\in[0,T]}\left\vert\int_0^t z_i(s)^{r-1}\, dW_i(s)\right\vert&\leq E\left\vert\int_0^T \mid z_{i}(s)\mid^{2(r-1)} \, ds \right\vert^{\frac{1}{2}}\nonumber\\
&\leq C E\left\vert\int_0^T (\mid z_{i}(s)\mid^{2r}+1)ds\right\vert^{\frac{1}{2}}\nonumber\\
&\leq C E\sup_{t\in[0,T]}\mid z_{i}(t)\mid^r T^\frac{1}{2}+CT^\frac{1}{2}\nonumber\\
&\leq CT^\frac{1}{2} \int_0^T E\sup_{s\in[0,t]}\mid z_{i}(s)\mid^r dt +CT^\frac{1}{2}.\nonumber
\end{align}
Putting everything together we obtain 
\begin{align*}
&\sum_{i=1}^m E\sup_{t\in[0,T]}\mid z_i(t)\mid^r + E\sup_{t\in[0,T]}\mid p(t)\mid^r\\
&\leq C\left\{ T+m\frac{\alpha}{\epsilon}T+m\sqrt{\frac{2T\alpha}{\epsilon}}\right\}\\
&+C(1+\frac{\lambda}{\sqrt{\epsilon}})\int_0^T
 E \sup_{s\in[0,t]}\mid p(s)\mid^r dt \\
&+C(\frac{\lambda}{\sqrt{\epsilon}}+\frac{\alpha}{\epsilon}+\sqrt{\frac{\alpha}{\epsilon}})\int_0^T \sum_{i=1}^m E \sup_{s\in[0,t]}\mid z_i(s)\mid^r dt\\
&+E\mid p(0)\mid^r+\sum_{i=1}^mE\mid z_i(0)\mid^r,
\end{align*}
so that
$$
\epsilon^{\frac{r}{2}}\xi(T)\leq C\epsilon^{\frac{r}{2}}+ C \epsilon^{\frac{r}{2}-1}(T+\sqrt{T})
+C \epsilon^{\frac{r}{2}-1}\int_0^T ds \xi(s).
$$
Applying Gronwall's Lemma gives then
$$
\epsilon^{\frac{r}{2}}\xi(T)\leq C(T+\sqrt{T})\epsilon^{\frac{r}{2}-1}e^{CT\epsilon^{\frac{r}{2}-1}}
$$
and, since $\,\epsilon^{\frac{r}{2}}\sum_{i=1}^m E\sup_{t\in[0,T]}\mid z_i(t)\mid^r\leq\epsilon^{\frac{r}{2}}\xi(T)$, we get the result.

\end{proof}

%
%
%%%%%%%%%%%%%%%%%%%%%%%%%%%%%%%%%%%%%%%%%%%%%%%%%%%%%%%%%%%%%%%%%%%%%%%%%%%%%%
%
%
\noindent {\bf Acknowledgments.} The authors are grateful to Dr Karel Pravda-Starov for many useful comments and for an extremely careful reading of an earlier version of the paper.

%%%%%%%%%%%%%%%%%%%%%%%%%%%%%%%%%%%%%%%%%%%%%%%%%%%%%%%%%%%%%%%%%%%%%%%%%%%%%%%

\appendix
\section{Hypocoercivity}
\label{app:hypoco}

In this appendix we recall some of the main results from the theory of hypocoercivity, as presented in~\cite{Vil04HPI}.
Throughout this Appendix we will use the notation introduced
in Section \ref{sec:notation}.
\begin{definition}[Coercivity]\label{def:coercivity}
Let $\mathcal{T}$ be an unbounded operator on a Hilbert space $\mathcal{H}$ with kernel $\mathcal{K}$. Let $\tilde{\mathcal{H}}$ be another Hilbert space continuously and densely embedded in $\mathcal{K}^{\perp}$. The operator $\mathcal{T}$ is said to be $\lambda$-coercive on $\tilde{\mathcal{H}}$ if
$$ 
(\mathcal{T}h,h)_{\tilde{\mathcal{H}}}\geq \lambda \|h\|^2_{\tilde{\mathcal{H}}} \quad \forall 
h\in\mathcal{K}^{\perp}\cap D(\mathcal{T}).
$$
\end{definition}
The following Proposition gives an equivalent Definition of coercivity.
\begin{prop}
With the same notation as in Definition \ref{def:coercivity}, $\mathcal{T}$ is $\lambda$-coercive on $\tilde{\mathcal{H}}$ iff
$$
\parallel e^{-\mathcal{T}t}h\parallel_{\tilde{\mathcal{H}}}\leq e^{-\lambda t}\parallel h\parallel_{\tilde{\mathcal{H}}}
\quad \forall h\in \tilde{\mathcal{H}} \mbox{ and }t\geq 0.
$$ 
\end{prop}
\begin{theorem}\label{thm:hypocoercivity}
 Let $\op$ be an operator of the form $\op=A^*A+B$, with $B^* = - B$, $\mathcal{K}=Ker\op$ and assume there exists $N\in\mathbb{N}$ such that
\begin{equation}
[C_{j-1},B]=C_j+R_j \quad 1\leq j\leq N+1, \quad C_0=A, C_{N+1}=0.
\end{equation}
Consider the following assumptions: for $k=0,...,N+1$
\begin{enumerate}
\item  $[A,C_k]$ is relatively bounded with respect to 
$\left\{C_j\right\}_{0\leq j\leq k}$ and 
$\left\{C_jA\right\}_{0\leq j\leq k-1}$ 
\item $[C_k,A^*]$ is relatively bounded with respect $I$ and 
$\left\{C_j \right\}_{0\leq j\leq k}$ (here $I$ indicates the identity operator on $\fspace$)
\item $R_k$ is relatively bounded with respect to  $\left\{C_j \right\}_{0\leq j\leq k-1}$ and $\left\{C_jA\right\}_{0\leq j\leq k-1}$
\item $\sum_{j=0}^N C_j^*C_j$ is $\kappa$-coercive for some $\kappa>0$ 
\end{enumerate}
If Assumptions $1-3$ are satisfied then there exists a scalar product $((\cdot ,\cdot ))$ on $H^1$ defining a norm equivalent to the usual $H^1$ norm and such that
\begin{equation}
\forall h\in H^1/\mathcal{K}, 
\quad ((h,\op h))\geq K \sum_{j=0}^N\|C_j h\|^2
\end{equation}
for some constant $K>0$. Furthermore, if Assumption $4$ is satisfied,  then there exists a constant $\lambda>0$ such that
$$
\forall h\in H^1/\mathcal{K},\quad ((h,\op h))\geq \lambda ((h,h)).
$$
In particular, $\op$ is hypocoercive in  $H^1/\mathcal{K}$, i.e.
$$
\|e^{-t\op}\|_{ H^1/\mathcal{K}\rightarrow  
H^1/\mathcal{K}}\leq C e^{-\lambda t}
$$
for some $C,\lambda>0$.
\end{theorem}
\begin{theorem}\label{thm:hypoeelipticity short time asymp}
With the same notation as in Theorem \ref{thm:hypocoercivity}, if Assumptions 1-3 are satisfied then 
$$
\|C_k e^{-t\op}h\|\leq \frac{C}{t^{k+\frac{1}{2}}}\|h\|\quad \forall k=0,...,N
$$
for all functions $h\in\fspace$.
\end{theorem}
\begin{theorem}\label{thm:long time entropic}
Let $V(x)\in C^2({\mathbb{R}^d})$ such that $\mu(dx)=e^{-V(x)}dx$ is a probability measure on $\mathbb{R}^d$ and assume that $\op$ generates a semigroup on a suitable space of positive functions. Let $\{A_j\}_{1\leq j \leq M}$ and $B$ be first order differential operators with smooth coefficients, with $B = -B^*$. Assume there exists $N\in\mathbb{N}$ such that
$$
[C_{j-1},B]=C_j+R_j \quad 1\leq j\leq N+1, \quad C_0=A, C_{N+1}=0.
$$
If, for $0\leq k \leq N+1$ the following assumptions are fulfilled 
\begin{enumerate}
\item 
$[A,C_k]$ is pointwise bounded with respect to A.
\item $[C_k, A^*]$ is pointwise bounded with respect to $I$ and $\left\{C_j \right\}_{0\leq j\leq k}$.
\item $R_k$ is pointwise bounded with respect to  $\left\{C_j \right\}_{0\leq j\leq k-1}$.
\item $[A,C_k]^*$ is pointwise bounded relatively to $I$ and $A$. 
\item there is a positive constant $\lambda>0$ such that $\sum_{k}C_k^*C_k\geq \lambda I$ pointwise on $\mathbb{R}^d$ ($I$ is the identity matrix on 
$\mathbb{R}^d$).
\item The probability measure $\mu$ satisfies a logaritmic Sobolev inequality.
\end{enumerate}
Then the Kullback information and the Fisher information decay exponentially fast to zero.
\end{theorem}
\begin{theorem}\label{thm:short time entropic}
With the same notation as in Theorem \ref{thm:long time entropic},
let $V(x)\in C^2({\mathbb{R}^d})$ be such that $\mu(dx)=e^{-V(x)}dx$ is a probability measure on $\mathbb{R}^n$ and assume that $\op$ generates a semigroup on a suitable space of positive functions.  If Assumptions $1-4$ of Theorem  \ref{thm:long time entropic} are fulfilled,
then, along the semigroup, the following bounds hold
$$
\int h_t \mid C_k \log h_t\mid^2 d\mu\leq\frac{C}{t^{2k+1}}\int h_0 \log h_0 d\mu\quad \forall k=0,...,N,
$$
where $h_t=f_t/\rho$ and $f_t$ is the law of the process with generator $L=-\op$.
\end{theorem}

\section{Ergodicity}\label{sec:app_ergodicity}

In this appendix we apply Markov chain techniques~\cite{MT, MatSt02, ReyBelletThomas2002} to prove ergodicity of the Markov process $\bx :=\{q(t), \, p(t), \, {\bf z}(t) \}$ given by~\eqref{syst_n_additional_proc}. To simplify the notation, we set all constants equal to $1$. We will consider the ergodic properties of the SDE
\begin{equation*}
\left\{
\begin{array}{ccl}
\dot{q}&=&p\\
\dot{p}&=&-\nabla_qV(q)+r\\
\dot{r}&=&-p-r+\dot{W}.
\end{array}
\right.
\end{equation*}  
We study either the case $q\in \mathbb{R}^{d}$ or $q\in \mathbb{T}^{d}$ and  $p,r \in \mathbb{R}^d$; we will sometimes use the notation $x=(q,p,r)$ and $L=-\op$ is the generator of the process.

Motivated by~\cite{MatSt02}, consider the following conditions:
\begin{description}
\item[Lyapunov Condition]:
 There exists a function $G(x):R^{3d}\rightarrow[1,\infty)$ such that $G(x)\rightarrow\infty$ as $\|x\|\rightarrow\infty$ and 
$LG(x)\leq -aG(x)+d$ for some \nolinebreak[4] $a,d>0$.
\item[Minorization condition]:
Let $P_t(x,A)$ be the transition kernel of the Markov process $\bx$. There exist $T>0, \eta>0$ and a probability measure $\nu$, with $\nu(C^c)=0$ and $\nu(C)=1$ for some fixed compact set $C$ in the phase space, such that 
$$
P_T(x,A)\geq\mu\nu(A) \quad\forall A\in\mathcal{B}(\mathbb{R}^{3d}),
\, x\in C. 
$$
\end{description}
We will use the following result, whose proof can be found in~\cite{MT}.
\begin{theorem}:
Minorization condition + Lyapunov condition $\Rightarrow\ $ ergodicity.
 \end{theorem}
\textbf{Assumption} $(\star)$: Let $B_s(y)\in R^{3d}$ be the ball of radius $s$ centered in $y$. For some fixed compact set $C$ we have
\begin{itemize}
\item  $P_t(x,A)$ has a density $p_t(x,y)$ which is continuous $\forall (x,y)\in C\times C$, more precisely
$$
P_t(x,A)=\int_{A}p_t(x,y) \, dy \quad \forall A\in\mathcal{B}(\mathbb{R}^{3d})\cap\mathcal{B}(C), \forall x\in C;
$$
\item $\forall x\in C$ and $\forall \delta>0$  one can find a $\bar{t}=\bar{t}(\delta)$ such that
$$
P_{\bar{t}}(x,B_{\delta}(x^{\ast}))>0\quad \mbox{for some } x^{\ast}\in int(C),\forall x\in C \mbox{.} 
$$
\end{itemize}
\begin{lemma}
Assumption ($\star$) $\Longrightarrow$ Minorization Condition.
\end{lemma}

Now let $V(q)$ be any $C^1(\mathbb{T}^{d})$ potential, $V(q)> -k$ for some positive constant $k$. Consider the function
$$
G(x)=\hat{C}+\frac{B}{2}\|p\|^2+\frac{C}{2}\|r\|^2+D V(q)+H(p,r),
$$
where $B,C,D,H$ and $\hat{C}$ are positive constants to be chosen. We have that
\begin{equation}\label{positivity of G(x)}
G(x)\geq \hat{C}+\frac{B}{2}\|p\|^2+\frac{C}{2}\|r\|^2-\frac{H}{2}\|p\|^2
-\frac{H}{2}\|r\|^2-Dk,
\end{equation}
so we need $B>H$, $C>H$ and $\hat{C}>Dk$. Moreover 
\begin{eqnarray*}
LG(x)&=&D(\nabla_qV,p)-B(\nabla_qV,p)-H(\nabla_qV,r)+B(r,p)+H \|r\|^2\\
& &-C(p,r)-H \|p\|^2-C \|r\|^2-H(p,r)+C\\
& &\leq H \|r\|^2+\frac{H}{4}\|\nabla_qV\|^2-H\|p\|^2-C\|r\|^2+C+H\|r\|^2,
\end{eqnarray*}
where we have chosen $B=D=C+H$. On the other hand, since $V(q)\leq K$,
$$
G(x)\geq -\frac{a}{2}B\|p\|^2-\frac{a}{2}C\|r\|^2-aKB-a\frac{H}{2}\|r\|^2-a\frac{H}{2}\|p\|^2
$$
so imposing also $2H-C\leq-\frac{a}{2}(C+H)$, $-H \leq-\frac{a}{2}(B+H)$ for some $a>0$, the Lyapunov condition is satisfied. One possible choice is $a=1/4$, $B=13/16$, $C=5/8$ and H=3/16.\\
Notice that from what we have just proven it follows that $\forall l\geq 1$ we have 
\begin{equation}\label{Lyapunov l>1}
L G(x)^l\leq-a_lG(x)^l+d_l,
\end{equation}
for some suitable positive constants $a_l$ and $d_l$. In fact,
$$
\partial_{q_i}G(x)^l=lG(x)^{l-1}\partial_{q_i}G(x),
$$
$$
\partial_{p_i}G(x)^l=lG(x)^{l-1}\partial_{p_i}G(x),
$$
and 
$$
\partial_{r_i}^2G(x)=\partial_{r_i}\left[lG(x)^{l-1}\partial_{r_i}G(x)\right]=l(l-1)G(x)^{l-2}(\partial_{r_i}G)^2
+l G(x)^{l-1}\partial^2_{r_i} G(x).
$$
Furthermore, using (\ref{positivity of G(x)}), we obtain
$$
l(l-1)G(x)^{l-2}(\partial_{r_i}G)^2\leq c_l G(x)^{l-1}, 
$$
so that 
$$
LG(x)^l\leq lG(x)^{l-1}LG(x)+c_lG(x)^{l-1}.
$$
Hence, using what we have proven in the case $l=1$, we get 
(\ref{Lyapunov l>1}).

Consider now the case $\bx\in\R^d\times\R^d\times\R^d$. We introduce the Lyapunov function
\begin{eqnarray*}
G(x)&=&\hat{C}+\frac{A}{2}\|q\|^2+\frac{B}{2}\|p\|^2+\frac{C}{2}\|r\|^2+DV(q)\\
& &+E(p,q)+F(q,r)+H(p,r)+M(\nabla_q V,p).
\end{eqnarray*}
Consequently, 
$$
\nabla_qG=Aq+D\nabla_qV+Ep+Fr+M \nabla^2V(q)\cdot p,
$$
$$
\nabla_pG=Bp+Eq+Hr+M \nabla_qV,
$$
$$
\nabla_rG=Cr+Fq+Hp.
$$
Thus, 
\begin{eqnarray*}
LG(x)&=&A(p,q)+D(\nabla_{q}V,p)+E \|p\|^2+F(p,r)-B(\nabla_qV,p)\\
& &-E(\nabla_qV,q)-H(\nabla_qV,r)+B(p,r)+E(q,r)+H\|r\|^2\\
& &-C(p,r)-F(p,q)-H\|p\|^2+M (p,HessV(q)\cdot p)\\
& &-C\|r\|^2-F(r,q)-H(p,r)+C-M\|\nabla_qV\|^2+M (r,\nabla_qV)
\end{eqnarray*}
From Assumption~\ref{assump:potential}(iii) it follows that there exist constants $\tilde{\beta}$ and $\tilde{\sigma}$ such that
$$ 
\tilde{\sigma}\|q\|^2-\tilde{\beta}\|\nabla_qV\|^2\rightarrow +\infty \;\mbox{as }\|q\|^2\rightarrow +\infty.
$$
Hence, it follows that $G$ satisfies the Lyapunov condition. Also in  this case, one can prove that the Lyapunov condition holds for $G(x)^l$, $l\geq 1$, as well.

As for Assumption ($\star$), first of all let us notice that, since the operator $\partial_t+\op$ is  hypoelliptic, the transition probability has a density; because the SDE we are considering has time independent coefficients the density is $C^{\infty}$ provided that $V(q)\in C^{\infty}$~\cite{Nualart2006}. Moreover,  studying the control problem associated with $dx=b(x)dt+\sigma dw$, namely $dX=b(X)dt+\sigma dU$ where $U(t)$ is a smooth control, and using Strook-Varadhan support Theorem, we can prove that $P_t(x,A)>0 \,\forall x\in \R^{3d}, \,t>0$ and for any open $A\in \R^{3d}$. 

Consider now the set $\mathcal{G}_l=\left\{g:\R^{3d}\rightarrow\R, \mbox{ measurable}: \mid g(x)\mid\leq V(x)^l\right\}$. We have proven that the process $x(t)$ has a unique invariant distribution $\rho$. Furthermore,  there exist constants $k=k(l)$ and $\lambda=\lambda(l)$, such that $\forall g\in \mathcal{G}_l$
\begin{equation}\label{geom_ergodicity}
\vert E^{x_0}g(x(t))-\rho(g)\vert \leq kG(x_0)e^{-\lambda t}, 
\quad t\geq 0.
\end{equation}
This completes the proof of the theorem.
%%%%%%%%%%%%%%%%%
%%%%%%%%%%%%%%%%%%%%%%%
%%%%%%%%%%%%%%%%%%%%%%%%%%%%%%%%%%%%%%%%%%%%%%%%%%%%%%%%%%%
%%%%%%%%%%%%%%%%%%%%%%%%%%%%%%%%%%%%%%%%%%%%%%%%%%%%%%%%%%%%%%%%%%%%%%%%%%%%%%%%%%%%%%%%%%%%%

\def\cprime{$'$} \def\cprime{$'$} \def\cprime{$'$} \def\cprime{$'$}
  \def\cprime{$'$} \def\cprime{$'$} \def\cprime{$'$}
  \def\Rom#1{\uppercase\expandafter{\romannumeral #1}}\def\u#1{{\accent"15
  #1}}\def\Rom#1{\uppercase\expandafter{\romannumeral #1}}\def\u#1{{\accent"15
  #1}}\def\cprime{$'$} \def\cprime{$'$} \def\cprime{$'$} \def\cprime{$'$}
  \def\cprime{$'$} \def\cprime{$'$} \def\cprime{$'$}
  \def\polhk#1{\setbox0=\hbox{#1}{\ooalign{\hidewidth
  \lower1.5ex\hbox{`}\hidewidth\crcr\unhbox0}}}

%\bibliography{../bibtex_files/mybib}
%\bibliographystyle{plain}

\end{document}